\documentclass{entcs}
\usepackage{amsmath,stmaryrd,amssymb,color,haskell,multicol,entcsmacro,tikz}
\usetikzlibrary{matrix,positioning,arrows}

\newcommand{\cat}[1]{\ensuremath{\mathbf{#1}}}
\newcommand{\op}{\ensuremath{{}^{\mathrm{op}}}}
\DeclareMathOperator{\arr}{arr}
\DeclareMathOperator{\first}{first}
\DeclareMathOperator{\second}{second}
\DeclareMathOperator{\inv}{inv}
\DeclareMathOperator{\fix}{fix}
\DeclareMathOperator{\dom}{dom}
\DeclareMathOperator{\cod}{cod}
\DeclareMathOperator{\Inv}{Inv}

\newcommand{\id}[1][]{\ensuremath{\mathrm{id}_{#1}}}
\newcommand{\prof}{\ensuremath{[\cat{C}\op\times\cat{C},\cat{Set}]}}

\newcommand{\acmp}{\ensuremath{\mathrel{>\!\!>\!\!>}}}
\newcommand{\afanout}{\ensuremath{\mathrel{\&\!\!\&\!\!\&}}}
\newcommand{\scolon}{\ensuremath{\mathrel{;}}}
\newcommand{\PInj}{\cat{PInj}}

\newcommand{\Pfn}{\cat{Pfn}}
\newcommand{\Kl}{\ensuremath{\mathcal{K}\!\ell}}
\newcommand{\ie}{\textit{i.e.}}
\newcommand{\eg}{\textit{e.g.}}

\newcommand{\changed}[1]{{\color{blue}#1}}
\newcommand{\cut}[1]{}

\begin{document}

\def\lastname{C. Heunen, R. Kaarsgaard, M. Karvonen}
\begin{frontmatter}
\title{Reversible effects as inverse arrows}
\author{Chris Heunen\thanksref{chrisemail}}
\address{School of Informatics\\ University of Edinburgh\\ United Kingdom}
\author{Robin Kaarsgaard\thanksref{robinemail}}
\address{Datalogisk Institut\\ University of Copenhagen\\ Denmark}
\author{Martti Karvonen\thanksref{marttiemail}}
\address{School of Informatics\\University of Edinburgh\\ United Kingdom}
\thanks[chrisemail]{Email:
    \href{mailto:chris.heunen@ed.ac.uk} {\texttt{\normalshape
        chris.heunen@ed.ac.uk}}}
\thanks[robinemail]{Email:
    \href{mailto:robin@di.ku.dk} {\texttt{\normalshape
        robin@di.ku.dk}}}
\thanks[marttiemail]{Email:
    \href{mailto:martti.karvonen@ed.ac.uk} {\texttt{\normalshape
        martti.karvonen@ed.ac.uk}}}
\begin{abstract}
  Reversible computing models settings in which all processes can be reversed.
  Applications include low-power computing, quantum computing, and robotics.
  It is unclear how to represent side-effects in this setting, because conventional methods need not respect reversibility.
  We model reversible effects by adapting Hughes' arrows to dagger arrows and inverse arrows.
  This captures several fundamental reversible effects, including serialization 
  and mutable store computations.
  Whereas arrows are monoids in the category of profunctors, dagger arrows are involutive monoids in the category of profunctors, and inverse arrows satisfy certain additional properties.
  These semantics inform the design of functional reversible programs supporting side-effects.
\end{abstract}
\end{frontmatter}

\section{Introduction}\label{sec:introduction}

Reversible computing studies settings in which all processes can be reversed: programs can be run backwards as well as forwards.
Its history goes back at least as far as 1961, when Landauer formulated his physical principle that logically irreversible manipulation of information costs work.
This sparked the interest in developing reversible models of computation as a means to making them more energy efficient.
Reversible computing has since also found applications in high-performance computing~\cite{schordanetal:parallel}, process calculi~\cite{cristescuetal:picalculus}, probabilistic computing~\cite{stoddartlynas:probabilistic}, quantum computing~\cite{selinger:dagger}, and robotics~\cite{schultzbordignonstoy:reconfiguration}.

There are various theoretical models of reversible computations. The most well-known ones are perhaps Bennett's reversible Turing machines~\cite{bennett:reversibility} and Toffoli's reversible circuit model~\cite{toffoli:reversible}. There are also various other models of reversible automata~\cite{morita:twoway,kutribwendlandt:automata} and combinator calculi~\cite{abramsky:structrc,jamessabry:infeff}.

We are interested in models of reversibility suited to functional programming languages. Functional languages are interesting in a reversible setting for two reasons. First, they are easier to reason and prove properties about, which is a boon if we want to understand the logic behind reversible programming. Second, they are not stateful by definition, which eases reversing programs.
It is fair to say that existing reversible functional programming languages~\cite{jamessabry:theseus,yokoyamaetal:rfun} still lack various desirable constructs familiar from the irreversible setting. 

Irreversible functional programming languages like Haskell naturally take semantics in categories. The objects interpret types, and the morphisms interpret functions. 
Functional languages are by definition not stateful, and their categorical semantics only models pure functions. However, sometimes it is useful to have non-functional side-effects, such as exceptions, input/output, or indeed even state. Irreversible functional languages can handle this elegantly using monads~\cite{moggi:monads} or more generally arrows~\cite{hughes:programmingarrows}.

A word on terminology. We call a computation $a \colon X \to Y$ \emph{reversible} when it comes with a specified partner computation $a^\dag \colon Y \to X$ in the opposite direction. This implies nothing about possible side-effects.
Saying that a computation is \emph{partially invertible} is stronger, and requires $a \circ a^\dag \circ a = a$. Saying that it is \emph{invertible} is even stronger, and requires $a \circ a^\dag$ and $a^\dag \circ a$ to be identities.
We call this partner of a reversible computation its \emph{dagger}. 
In other words, reversible computing for us concerns dagger arrows on dagger categories, and is modeled using involutions~\cite{heunenkarvonen:daggermonads}. 
In an unfortunate clash of terminology, categories of partially invertible maps are called inverse categories~\cite{cockettlack:restrictioncategories}, and categories of invertible maps are called groupoids~\cite{gabbaykropholler:lazy}. 
Thus, inverse arrows on inverse categories concern partially invertible maps.

We develop \emph{dagger arrows} and \emph{inverse arrows}, which are useful in two ways:
\begin{itemize}
  \item We illustrate the reach of these notions by exhibiting many fundamental reversible computational side-effects that are captured (in Section~\ref{sec:arrows}), including: pure reversible functions, information effects, reversible state, serialization, vector transformations, 
  dagger Frobenius monads~\cite{heunenkarvonen:reversiblemonads,heunenkarvonen:daggermonads}, recursion~\cite{kaarsgaardaxelsengluck:joininversecategories}, and superoperators. Because there is not enough space for much detail, we treat each example informally from the perspective of programming languages, but formally from the perspective of category theory.   
  \item We prove that these notions behave well mathematically (in Section~\ref{sec:arrowscategorically}): whereas arrows are monoids in a category of profunctors~\cite{jacobs2009categorical}, dagger arrows and inverse arrows are involutive monoids.
\end{itemize}

This paper aims to inform design principles of sound reversible programming languages. 
The main contribution is to match desirable programming concepts to precise category theoretic constructions.
As such, it is written from a theoretical perspective. To make examples more concrete for readers with a more practical background, we adopt the syntax of a typed first-order reversible functional programming language with type classes. 
We begin with preliminaries on reversible base categories (in Section~\ref{sec:inversecategories}).

\section{Dagger categories and inverse categories}\label{sec:inversecategories}

This section introduces the categories we work with to model pure computations: dagger categories and inverse categories. Each has a clear notion of reversing morphisms. Regard morphisms in these base categories as pure, ineffectful maps.

\begin{definition}
  A \emph{dagger category} is a category equipped with a \emph{dagger}: a contravariant endofunctor $\cat{C} \to \cat{C}$ satisfying $f^{\dag\dag}=f$ for morphisms $f$ and $X^\dag=X$ for objects $X$. 
  A morphism $f$ in a dagger category is:
  \begin{itemize}
    \item \emph{positive} if $f=g^\dag \circ g$ for some morphism $g$;
    \item a \emph{partial isometry} if $f = f \circ f^\dag \circ f$;
    \item \emph{unitary} if $f \circ f^\dag = \id$ and $f^\dag \circ f = \id$.
  \end{itemize}
  A dagger functor is a functor between dagger categories that preserves the dagger, \ie\ a functor $F$ with $F(f^\dag)=F(f)^\dag$.
  A \emph{(symmetric) monoidal dagger category} is a monoidal category equipped with a dagger making the coherence isomorphisms 
    \begin{align*}
    &\alpha_{X,Y,Z}\colon X\otimes (Y\otimes Z)\to (X\otimes Y)\otimes Z \qquad \rho_X\colon X\otimes I\to X \\ 
    &\lambda_X\colon I\otimes X\to X \qquad \text{(and }\sigma_{X,Y}\colon X\otimes Y\to Y\otimes X\text{ in the symmetric case)}
    \end{align*}
    unitary and satisfying $(f \otimes g)^\dag = f^\dag \otimes g^\dag$ for morphisms $f$ and $g$.
  We will sometimes suppress coherence isomorphisms for readability.
\end{definition}

Any groupoid is a dagger category under $f^\dag=f^{-1}$. Another example of a dagger category is $\cat{Rel}$, whose objects are sets, and whose morphisms $X \to Y$ are relations $R \subseteq X \times Y$, with composition $S \circ R = \{ (x,z) \mid \exists y \in Y \colon (x,y) \in R, (y,z) \in S\}$. The dagger is $R^\dag = \{(y,x) \mid (x,y) \in R\}$. It is a monoidal dagger category under either Cartesian product or disjoint union.

\begin{definition}
   A \emph{(monoidal) inverse category} is a (monoidal) dagger category of partial isometries where positive maps commute: $f \circ f^\dag \circ f=f$ and $f^\dag \circ f \circ g^\dag \circ g = g^\dag\circ g \circ f^\dag\circ f$ for all maps $f \colon X \to Y$ and $g \colon X \to Z$. 
\end{definition}

Every groupoid is an inverse category. Another example of an inverse category is $\cat{PInj}$, whose objects are sets, and morphisms $X \to Y$ are partial injections: $R \subseteq X \times Y$ such that for each $x \in X$ there exists at most one $y \in Y$ with $(x,y) \in R$, and for each $y \in Y$ there exists at most one $x \in X$ with $(x,y) \in R$. It is a monoidal inverse category under either Cartesian product or disjoint union.

\begin{definition}
  A dagger category is said to have \emph{inverse products}~\cite{giles:thesis} if it is a symmetric monoidal dagger category with a natural transformation $\Delta_X \colon X \to X
  \otimes X$ making the following diagrams commute:
  \[
    \begin{aligned}\begin{tikzpicture}
        \matrix (m) [matrix of math nodes,row sep=2em,column sep=4em,minimum width=2em]
        {X  & X\otimes X \\
          & X\otimes X \\};
        \path[->]
        (m-1-1) edge node [below] {$\Delta_X$} (m-2-2)
            edge node [above] {$\Delta_X$} (m-1-2)
        (m-1-2) edge node [right] {$\sigma_{X,X}$} (m-2-2);
    \end{tikzpicture}\end{aligned}
    \begin{aligned}\begin{tikzpicture}
        \matrix (m) [matrix of math nodes,row sep=2em,column sep=3em,minimum width=2em]
        {X  & & X\otimes X  \\
        X\otimes X & X\otimes (X\otimes X) & (X\otimes X)\otimes X \\};
        \path[->]
        (m-1-1) edge node [right] {$\Delta_X$} (m-2-1)
            edge node [above] {$\Delta_X$} (m-1-3)
        (m-2-1) edge node [below] {$\id\otimes\Delta_X$} (m-2-2)
        (m-2-2) edge node [below] {$\alpha$} (m-2-3)
        (m-1-3) edge node [left] {$\Delta_X\otimes\id$} (m-2-3);
    \end{tikzpicture}\end{aligned}
  \]
  \[
    \begin{aligned}\begin{tikzpicture}
        \matrix (m) [matrix of math nodes,row sep=2em,column sep=4em,minimum width=2em]
        {X  & X\otimes X \\
          & X \\};
        \path[->]
        (m-1-1) edge node [below] {$\id$} (m-2-2)
            edge node [above] {$\Delta_X$} (m-1-2)
        (m-1-2) edge node [right] {$\Delta_X^\dag$} (m-2-2);
    \end{tikzpicture}\end{aligned}
    \begin{aligned}\begin{tikzpicture}
        \matrix (m) [matrix of math nodes,row sep=1.5em,column sep=4em,minimum width=2em]
        {X\otimes X  && X\otimes(X\otimes X) \\
         & X& \\
         (X\otimes X)\otimes X & & X\otimes X \\};
        \path[->]
        (m-1-1) edge node [right] {$\Delta\otimes\id$} (m-3-1)
            edge node [above] {$\id\otimes\Delta_X$} (m-1-3)
            edge node [right=3mm] {$\Delta_X^\dag$} (m-2-2)
        (m-2-2) edge node [below] {$\Delta_X$} (m-3-3)
        (m-3-1) edge node [below] {$(\id\otimes\Delta_X^\dag)\circ\alpha^\dag$} (m-3-3)
        (m-1-3) edge node [left,yshift=4mm] {$(\Delta_X^\dag\otimes \id)\circ \alpha$} (m-3-3);
    \end{tikzpicture}\end{aligned}
  \]
  These diagrams express cocommutativity, coassociativity, speciality and the Frobenius law. 
\end{definition} 

Another useful monoidal product, here on inverse categories, is a disjointness tensor, defined in the following way (see \cite{giles:thesis}):

\begin{definition}\label{def:disjtensor}
  An inverse category is said to have a \emph{disjointness tensor} if it is
  equipped with a symmetric monoidal tensor product $- \oplus -$ such that its
  unit $0$ is a zero object, and the canonical \emph{quasi-injections}
  \begin{equation*}
    \amalg_1 = X \xrightarrow{\rho^{-1}_X} X \oplus 0 \xrightarrow{X \oplus 
    0_{0,Y}} X \oplus Y \qquad
    \amalg_2 = Y \xrightarrow{\lambda^{-1}_Y} 0 \oplus Y \xrightarrow{0_{0,X}
    \oplus Y} X \oplus Y
  \end{equation*}
  are jointly epic.
\end{definition}

For example, $\cat{PInj}$ has inverse products $\Delta_X \colon X \to X \otimes
X$ with $x \mapsto (x,x)$, and a disjointness tensor where $X \oplus Y$ is
given by the tagged disjoint union of $X$ and $Y$ (the unit of which is
$\emptyset$).

Inverse categories can also be seen as certain instances of restriction categories. Informally, a restriction category models partially defined morphisms, by assigning to each $f\colon A\to B$ a morphism $\bar{f}\colon A\to A$ that is the identity on the domain of definition of $f$ and undefined otherwise. For more details, see~\cite{cockettlack:restrictioncategories}.

\begin{definition}\label{def:restrictioncategory}
  A \emph{restriction category} is a category equipped with an operation that assigns to each $f\colon A\to B$ a morphism $\bar{f}\colon A\to A$ such that:
  \begin{itemize} 
    \item $f\circ\bar{f}=f$ for every $f$;
    \item $\bar{f}\circ\bar{g}=\bar{g}\circ\bar{f}$ whenever $\dom f=\dom g$;
    \item $\overline{g\circ\bar{f}}=\bar{g}\circ\bar{f}$ whenever $\dom f=\dom g$;
    \item $\bar{g}\circ f=f\circ\overline{g\circ f}$ whenever $\dom g=\cod f$.
  \end{itemize}
  A \emph{restriction functor} is a functor $F$ between restriction categories with $F(\bar{f})=\overline{F(f)}$. A \emph{monoidal restriction category} is a restriction category with a monoidal structure for which $\otimes\colon \cat{C}\times\cat{C}\to\cat{C}$ is a restriction functor.

  A morphism $f$ in a restriction category is a \emph{partial isomorphism} if there is a morphism $g$ such that $g \circ f=\bar{f}$ and $f \circ g=\bar{g}$. Given a restriction category \cat{C}, define $\Inv(\cat{C})$ to be the wide subcategory of \cat{C} having all partial isomorphisms of \cat{C} as its morphisms.
\end{definition}

An example of a monoidal restriction category is $\cat{PFn}$, whose objects are sets, and whose morphisms $X \to Y$ are partial functions: $R \subseteq X \times Y$ such that for each $x \in X$ there is at most one $y \in Y$ with $(x,y) \in R$. The restriction $\bar{R}$ is given by $\{(x,x) \mid \exists y \in Y \colon (x,y) \in R\}$.

\begin{remark}
  Inverse categories could equivalently be defined as either categories in which every morphism $f$ satisfies $f=f \circ g \circ f$ and $g=g \circ f \circ g$ for a unique morphism $g$, or as restriction categories in which all morphisms are partial isomorphisms~\cite[Theorem~2.20]{cockettlack:restrictioncategories}. It follows that functors between inverse categories automatically preserve daggers and that $\Inv(\cat{C})$ is an inverse category.

  It follows, in turn, that an inverse category with inverse products is a monoidal inverse category: because $X \otimes -$ and $- \otimes Y$ are endofunctors on an inverse category, they preserve daggers, so that by bifunctoriality $-\otimes -$ does as well:
 \[
   (f \otimes g)^\dag 
   = ((f \otimes \id[Y]) \circ (\id[X] \otimes g))^\dag 
   = (\id[X] \otimes g)^\dag \circ (f \otimes \id[Y])^\dag 
   = (\id[X] \otimes g^\dag) \circ (f^\dag \otimes \id[Y]) 
   = f^\dag \otimes g^\dag\text.
 \]
\end{remark}

\section{Arrows as an interface for reversible effects}\label{sec:arrows}

Arrows are a standard way to encapsulate computational side-effects in a functional (irreversible) programming language~\cite{hughes:arrows,hughes:programmingarrows}. This section extends the definition to reversible settings, namely to dagger arrows and inverse arrows. We argue that these notions are ``right'', by exhibiting a large list of fundamental reversible side-effects that they model.
We start by recalling irreversible arrows.

\begin{definition}\label{def:arrow}
  An \emph{arrow} on a symmetric monoidal category $\cat{C}$ is a functor $A \colon \cat{C}\op \times \cat{C} \to \cat{Set}$ with operations
  \begin{align*}
    \arr &: (X \to Y)\to A~X~Y \\
    (\acmp) & : A~X~Y \to A~Y~Z \to A~X~Z\\
    \first_{X,Y,Z} &: A~X~Y \to A~(X\otimes Z)~(Y\otimes Z) 
  \end{align*}
  that satisfy the following laws:
  \begin{align}
    (a \acmp b) \acmp c &= a \acmp (b \acmp c) \label{eq:arrow1}\\
    \arr (g\circ f) &= \arr f \acmp \arr g \label{eq:arrow2}\\
    \arr \id \acmp a =&\;a = a \acmp\arr \id  \label{eq:arrow3}\\
    \first_{X,Y,I} a \acmp \arr \rho_Y &= \arr \rho_X \acmp a  \label{eq:arrow4}\\
    \first_{X,Y,Z} a \acmp \arr (\id[Y]\otimes f) &=  \arr(\id[X]\otimes f) \acmp \first_{X,Y,Z} a \label{eq:arrow5}\\
    (\first_{X,Y,Z\otimes V} a) \acmp \arr \alpha_{Y,Z,V} &= \arr \alpha_{X,Z,V} \acmp \first (\first a) \label{eq:arrow6}\\
    \first (\arr f)&= \arr (f\otimes\id) \label{eq:arrow7}\\
    \first (a \acmp b)&=(\first a) \acmp (\first b) \label{eq:arrow8}
  \end{align}
  where we use the functional programming convention to write $A~X~Y$ for $A(X,Y)$ and $X\to Y$ for $\hom (X,Y)$
  The \emph{multiplicative fragment} consists of above data except $\first$, satisfying all laws except those mentioning $\first$; we call this a \emph{weak arrow}.

  Define $\second(a)$ by $\arr(\sigma) \acmp \first(a) \acmp \arr(\sigma)$, using the symmetry, so analogs of~\eqref{eq:arrow4}--\eqref{eq:arrow8} are satisfied. Arrows makes sense for (nonsymmetric) monoidal categories if we add this operation and these laws.
\end{definition}

Here come our central definitions.

\begin{definition}
  A \emph{dagger arrow} is an arrow on a monoidal dagger category with an additional operation
  $\inv : A~X~Y \to A~Y~X$
  satisfying the following laws:
  \begin{align}
    \inv (\inv a)&= a \label{eq:daggerarrow1}\\
    \inv a \acmp \inv b &= \inv (b\acmp a) \label{eq:daggerarrow2}\\ 
    \arr (f^\dag)&=\inv (\arr f) \label{eq:daggerarrow3}\\
    \inv (\first a )&=\first (\inv a) \label{eq:daggerarrow4}
  \intertext{An \emph{inverse arrow} is a dagger arrow on a monoidal inverse category such that:}
    (a \acmp \inv a) \acmp a &= a \label{eq:inversearrow1}\\
    (a \acmp \inv a) \acmp (b \acmp \inv b) &= (b \acmp \inv b) \acmp (a \acmp \inv a) \label{eq:inversearrow2}
  \end{align}
  The \emph{multiplicative fragment} consists of above data except $\first$, satisfying all laws except those mentioning $\first$.
\end{definition}

\begin{remark}
  There is some redundancy in the definition of an inverse arrow:
  \eqref{eq:inversearrow1} and~\eqref{eq:inversearrow2} imply \eqref{eq:daggerarrow3} and~\eqref{eq:daggerarrow4}; and~\eqref{eq:daggerarrow3} implies $\inv (\arr \id) = \arr \id$.   
\end{remark}

Like the arrow laws~\eqref{eq:arrow1}--\eqref{eq:arrow8}, in a programming language with inverse arrows, the burden is on the programmer to guarantee~\eqref{eq:daggerarrow1}--\eqref{eq:inversearrow2} for their implementation. If that is done, the language guarantees arrow inversion.

\begin{remark}
  Now follows a long list of examples of inverse arrows, described in a typed first-order reversible functional pseudocode with type classes, inspired by Theseus~\cite{jamessabry:theseus,jamessabry:infeff}, the revised version of Rfun (briefly described in~\cite{kaarsgaardthomsen:rfun}), and Haskell. Type classes are a form of interface polymorphism: A type class is defined by a class specification containing the signatures of functions that a given type must implement in order to be a member of that type class (often, type class membership also informally requires the programmer to ensure that certain equations are required of their implementations). For example, the \<Functor\> type class (in Haskell) is given by the class specification

\begin{haskell}
\hskwd{class} Functor f \hskwd{where} \\
\quad\hsalign{
  fmap &:&\ (a\to{b})\to{f a}\to{f b}
}
\end{haskell}
with the additional informal requirements that \<fmap id = id\> and \<fmap (g\circ{f}) = (fmap g)\circ(fmap f)\> must be satisfied for all instances. For example, lists in Haskell satisfy these equations when defining \<fmap\> as the usual \<map\> function, \ie:
\begin{haskell*}
  \hskwd{instance} Functor List \hskwd{where} \\
  \quad\hsalign{
    fmap &:&\ (a\to{b})\to{List a}\to{List b} \\
    fmap f [] &=&\ [] \\
    fmap f (x{::}xs) &=&\ (f x){::}(fmap f xs)
  }
\end{haskell*}

While higher-order reversible functional programming is fraught, aspects of this can be mimicked by means of parametrized functions. A parametrized function is a function that takes parts of its input statically (\ie, no later than at compile time), in turn lifting the first-order requirement on these inputs. To separate static and dynamic inputs from one another, two distinct function types are used: $a \to b$ denotes that $a$ must be given statically, and $a \leftrightarrow b$ (where $a$ and $b$ are first-order types) denotes that $a$ is passed dynamically. As the notation suggests, functions of type $a \leftrightarrow b$ are reversible. For example, a parametrized variant of the reversible map function can be defined as a function \<map : (a\leftrightarrow{b})\to([a]\leftrightarrow[b])\>. Thus, \<map\> itself is \emph{not} a reversible function, but given statically any reversible function \<f : a\leftrightarrow{b}\>, the parametrized \<map f : ([a]\leftrightarrow[b])\> is.

  Given this distinction between static and dynamic inputs, the signature of $\arr$ becomes $(X \leftrightarrow Y) \to A~X~Y$. 
  We will see later that Arrows on $\cat{C}$ can be modelled categorically as monoids in the functor category $\prof$~\cite{jacobs2009categorical}.
  Definition~\ref{def:arrow} uses the original signature, because this distinction is not present in the irreversible case. Fortunately, the semantics of arrows remain the same whether or not this distinction is made.
\end{remark}

\begin{example}{\emph{(Pure functions)}}\label{ex:pure}
  A trivial example of an arrow is the identity arrow $\hom(-,+)$ which adds
  no computational side-effects at all. This arrow is not as boring as it may look at first. If the identity arrow is an inverse arrow, then the programming language 
  in question is both \emph{invertible} and \emph{closed under program 
  inversion}: any program $p$ has a semantic inverse $\llbracket p \rrbracket^\dag$ (satisfying certain equations), and the semantic inverse coincides with the
  semantics $\llbracket \inv(p) \rrbracket$ of another program $\inv(p)$. As 
  such, $\inv$ must be a sound and complete \emph{program inverter} (see also 
  \cite{kawabeglueck:lrinv}) on pure functions; not a trivial matter at all.
\end{example}

\begin{example}{\emph{(Information effects)}}\label{ex:informationeffects}
  James and Sabry's \emph{information effects}~\cite{jamessabry:infeff} explicitly expose creation and erasure of information as effects.
  This type-and-effect system captures irreversible computation inside a pure reversible setting.

  We describe the languages from~\cite{jamessabry:infeff} categorically, as there is no space for syntactic details.
  Start with the free symmetric bimonoidal category $(\cat{C},\oplus,\otimes,0,1)$ where $\otimes$ and $\oplus$ satisfy the axioms of a commutative semiring up to coherent isomorphism, see~\cite{kelly1974coherence,laplaza1972coherence} for more details.
  Objects interpret types of the reversible language $\Pi$ of bijections, and morphisms interpret terms.
  It turns out that the category $(\cat{C},\otimes,1)$ is then automatically a monoidal inverse category.

  The category $\cat{C}$ carries an arrow, where $A(X,Y)$ is the disjoint union of $\hom(X \otimes H, Y \otimes G)$ where $G$ and $H$ range over all objects, and morphisms $X \otimes H \to Y \otimes G$ and $X \otimes H' \to Y \otimes G'$ are identified when they are equal up to coherence isomorphisms.
  This is an inverse arrow, where $\inv(a)$ is simply $a^{-1}$. 
  It supports the following additional operations:
  \begin{haskell}
    erase &=&\ [ \sigma \colon X \otimes 1 \to 1\otimes X ]_\simeq &  \in A(X,1)\text, \\
    create_X &=&\ [  \sigma\colon 1\otimes X \to X \otimes 1 ]_\simeq & \in A(1,X)\text. 
  \end{haskell}
  James and Sabry show how a simply-typed first order functional irreversible language translates into a reversible one by using this inverse arrow to build implicit communication with a global heap $H$ and garbage dump $G$.
\end{example}

\begin{example}{\emph{(Reversible state)}}\label{ex:reversiblestate}
  Perhaps the prototypical example of an effect is computation with a mutable 
  store of type $S$. In the irreversible case, such computations are performed
  using the state monad \<State S X = S \multimap (X \otimes 
  S)\>, where $S \multimap -$ is the right adjoint to $- \otimes S$, and can
  be thought of as a function type. Morphisms in the
  corresponding Kleisli category are morphisms of the form $X \to S \multimap
  (Y \otimes S)$ in the ambient monoidal closed category. In this 
  formulation, the current state is fetched by \<get : 
  State S S\> defined as \<get s = (s,s)\>, while the state is (destructively) 
  updated by \<put : S\to\mathit{State} S 1\> defined as \<put x s = ((),x)\>.
  
  Such arrows  can not be used as-is in inverse categories, however, as
  canonical examples (such as \PInj{}) fail to be monoidal closed.
  To get around this, note that it follows from monoidal closure that $\hom(X, S
  \multimap (Y \otimes S)) \simeq \hom(X \otimes S, Y \otimes S)$, so that
  $\hom(-\otimes S, -\otimes S)$ is an equivalent arrow that does not depend on
  closure. With this is mind, we define the \emph{reversible state arrow} with a
  store of type $S$:
  \begin{haskell*}
    \hskwd{type} RState S X Y = 
    X\otimes\mathit{S}\leftrightarrow\mathit{Y}\otimes\mathit{S} \\\\\\
    \hskwd{instance} Arrow (RState S) \hskwd{where} \\
    \quad\hsalign{
    arr f (x,s) &=&\ (f x, s) \\
    (a\acmp\mathit{b}) (x,s) &=&\ b (a (x,s)) \\
    first a ((x,z),s) &=&\ \hskwd{let} (x',s')=a (x,s) \hskwd{in} ((x',z),s') 
    } 
    \\\\\\
    \hskwd{instance} InverseArrow (RState S) \hskwd{where} \\
    \quad\hsalign{
    inv a (y,s) &=&\ a^\dagger (y,s)
    }
  \end{haskell*}
  This satisfies the inverse arrow laws. To access the state, we use reversible 
  duplication of values (categorically, this requires the monoidal product to 
  have a natural diagonal $\Delta_X : X \to X \otimes X$, as inverse products 
  do). Syntactically, this corresponds to the following arrow:
  \begin{haskell}
    get &:&\ RState S X (X\otimes\mathit{S}) \\
    get (x,s) &=&\ ((x,s),s)
  \end{haskell}
  The inverse to this arrow is 
  \<assert : RState S (X\otimes\mathit{S}) X\>, 
  which asserts that the current state is
  precisely what is given in its second input component; if this fails, the
  result is undefined. For changing the state, while we cannot destructively
  update it reversibly, we \emph{can} reversibly update it by a given
  reversible function with signature $S \leftrightarrow S$. This gives:
  \begin{haskell}
    update &:&\ (S\leftrightarrow\mathit{S})\to\mathit{RState} S X X \\
    update f (x,s) &=&\ (x, f s)
  \end{haskell}
  This is analogous to how variable assignment works in the reversible 
  programming language Janus~\cite{yokoyamaglueck:janus}: Since destructive
  updating is not permitted, state is updated by means of built-in
  reversible update operators, \eg, updating a variable by adding a constant 
  or the contents of another variable to it, etc.
  
  \cut{$\mathtt{x}~{+}{=}~\mathtt{y}$ (add
  the contents $\mathtt{y}$ to the contents of $\mathtt{x}$),
  $\mathtt{y}~{-}{=}~\mathtt{2}$ (subtract two from the contents of 
  $\mathtt{y}$), and so on.}


\end{example}

\begin{example}{\emph{(Computation in context)}}
  Related to computation with a mutable store is computation with an immutable
  one; that is, computation within a larger context that remains invariant
  across execution. In an irreversible setting, this job is typically handled
  by the \emph{reader monad} (with context of type $C$), defined as \<Reader C
  X = C\rightarrow{X}\>. This approach is fundamentally irreversible, however,
  as the context is ``forgotten'' whenever a value is computed by supplying it
  with a context. Even further, it relies on the reversibly problematic
  notion of monoidal closure.
  
  A reversible version of this idea is one that remembers the context, 
  giving us the reversible Reader arrow:
  \begin{haskell}
    \hskwd{type} Reader C X Y = X\otimes{C}\leftrightarrow\mathit{Y}\otimes{C}
  \end{haskell}
  This is precisely the same as the state arrow -- indeed, the instance
  declarations for \<arr\>, \<(\acmp)\>, \<first\>, and \<inv\> are the same --
  save for the fact that we additionally require \emph{all} \<Reader\> arrows
  $r$ to satisfy \<c=c'\> whenever \<r (x,c) = (y,c')\>. We notice that \<arr
  f\> satisfies this property for all \<f\>, whereas \<(\acmp)\>, \<first\>, and
  \<inv\> all preserve it. This resembles the ``slice'' construction on inverse
  categories with inverse products; see \cite[Sec.~4.4]{giles:thesis}.
  
  As such, while we can provide access to the context via a function defined
  exactly as \<get\> for the reversible state arrow, we cannot provide an
  update function without (potentially) breaking this property -- as intended.
  In practice, the property that the context is invariant across execution can
  be aided by appropriate interface hiding, \ie\ exposing the \<Reader\> type
  and appropriate instance declarations and helpers (such as \<get\> and
  \<assert\>) but leaving the \emph{constructor} for \<Reader\> arrows hidden.
\end{example}

\begin{example}{\emph{(Rewriter)}}
  A particularly useful special case of the reversible state arrow is when the
  store $S$ forms a group. \cut{In the irreversible case, this is related to
  the \emph{Writer} monad typically used to perform tasks such as logging.}
  While group multiplication if seen as a function
  \<G\otimes{G}\leftrightarrow{G}\> is invertible only in degenerate cases, we
  can use parametrization to fix the first argument of the multiplication,
  giving it a much more reasonable signature of
  \<G\to({G}\leftrightarrow{G})\>. In this way, groups can be expressed as
  instances of the type class
  \begin{haskell}
    \hskwd{class} Group G \hskwd{where} \\
    \quad\hsalign{
      gunit &:&\ G\\
      gmul &:&\ G\to({G}\leftrightarrow{G}) \\
      ginv &:&\ G\leftrightarrow{G}
    }
  \end{haskell}
  subject to the usual group axioms. This gives us an arrow of the form
  \begin{haskell}
    \hskwd{type} Rewriter G X Y = X\otimes{G}\leftrightarrow{Y}\otimes{G}
  \end{haskell}
  with instance declarations identical to that of \<RState G\>, save that we 
  require \<G\> to be an instance of the \<Group\> type class.
  With this, adding or removing elements from state of type $G$ can then be
  performed by
  \begin{haskell}
    rewrite &:&\ G\to{Rewriter} G X X\\
    rewrite a (x, b) &=&\ (x, gmul a b)
  \end{haskell}
  which ``rewrites'' the state by the value $a$ of type $G$. Note that while
  the name of this arrow was chosen to be evocative of the \emph{Writer} monad
  known from irreversible functional programming, as it may be used for similar
  practical purposes, its construction is substantially different (\ie,
  irreversible \emph{Writer} arrows are maps of the form $X \to Y \times M$
  where $M$ is a monoid).
  \cut{filling a similar role as the \<tell\> function does for the \<Writer\>
  monad in the irreversible case.}
\end{example}

\begin{example}{\emph{(Vector transformation)}}
  Vector transformations, that is, functions on lists that preserve the length
  of the list, form another example of inverse arrows. The \<Vector\> arrow is 
  defined as follows:
  \begin{haskell*}
    \hskwd{type} Vector X Y = 
    [X]\leftrightarrow[Y]\\\\\\
    \hskwd{instance} Arrow (Vector) \hskwd{where} \\
    \quad\hsalign{
    arr f xs &=&\ map f xs \\
    (a\acmp{b}) xs &=&\ b (a xs) \\
    first a ps &=&\ \hskwd{let} (xs,zs) = zip^\dagger ps \hskwd{in} zip (a xs,
    zs)
    } \\\\\\
    \hskwd{instance} InverseArrow (Vector) \hskwd{where} \\
    \quad\hsalign{
    inv a ys &=&\ a^\dagger ys
    }
  \end{haskell*}
  The definition of \<first\> relies on the usual \<map\> and \<zip\> 
  functions, which are defined as follows:\\
\begin{tabular}{ll}
\begin{minipage}{.4\textwidth}
\begin{haskell}
map &:&\ (a\leftrightarrow{b})\to([a]\leftrightarrow[b])\\
map f [] &=&\ []\\
map f (x{::}xs) &=&\ (f x){::}(map f xs)
\end{haskell}
\end{minipage}
&
\begin{minipage}{.4\textwidth}
\begin{haskell}
zip &:&\ ([a],[b])\leftrightarrow[(a,b)]\\
zip ([],[]) &=&\ []\\
zip (x{::}xs,y{::}ys) &=&\ (x,y){::}(zip (xs,ys))\\
\end{haskell}
\end{minipage}
\end{tabular}\\
  Notice that preservation of length is required for \<first\> to work: if the 
  arrow \<a\> does not preserve the length of \<xs\>, then \<zip (a xs, 
  zs)\> is undefined. However, since \<arr\> lifts a pure function $f$ to 
  a \<map\> (which preserves length), and \<(\acmp)\> and \<inv\> are 
  given by the usual composition and inversion, the interface maintains this property.
\end{example}

\begin{example}{\emph{(Reversible error handling)}}
  An inverse \emph{weak arrow} comes from reversible computation
  with a possibility for failure. The weak \<Error\> arrow is defined using disjointness tensors as
  follows:
  \begin{haskell*}
    \hskwd{type} Error E X Y = 
    X\oplus{E}\leftrightarrow{Y}\oplus{E}\\\\\\
    \hskwd{instance} WeakArrow (Error E) \hskwd{where} \\
    \quad\hsalign{
    arr f (InL x) &=&\ InL (f x) \\
    arr f (InR e) &=&\ InR e \\
    (a\acmp{b}) x &=&\ b (a x) \\
    } \\\\\\
    \hskwd{instance} InverseWeakArrow (Error E) \hskwd{where} \\
    \quad\hsalign{
    inv a y &=&\ a^\dagger y
    }
  \end{haskell*}
  In this definition, we think of the type \<E\> as the type of \emph{errors}
  that could occur during computation. As such, a pure function $f$ lifts to a
  weak arrow which always succeeds with value $f(x)$ when given a
  nonerroneous input of $x$, and always propagates errors that may have occured
  previously.
  
  Raising an error reversibly requires more work than in the irreversible case,
  as the effectful program that produces an error must be able to
  \emph{recover} from it in the converse direction. In this way, a reversible
  \<raise\> requires two pieces of data: a function of type
  \<X\leftrightarrow{E}\> that transforms problematic inputs into appropriate
  errors; and a choice function of type \<E\leftrightarrow{E}\oplus{E}\> that
  decides if the error came from this site, injecting it to the left if it did,
  and to the right if it did not. The latter choice function is critical, as in 
  the converse direction it decides whether the error should be handled 
  immediately or later. Thus we define \<raise\> as follows:
  \begin{haskell}
    raise &:&\ 
    (X\leftrightarrow{E})\to(E\leftrightarrow{E}\oplus{E})\to{Error} E X Y \\
    raise f p x &=&\ InR (p^\dagger (arr f x)))
  \end{haskell}
  The converse of \<raise\> is \<handle\>, an (unconditional) error handler that
  maps matching errors back to succesful output values. Since unconditional
  error handling is seldom required, this can be combined with control flow
  (see Example~\ref{ex:control_flow}) to perform conditional error handling,
  \ie\ to only handle errors if they occur.
\end{example}

\begin{example}{\emph{(Serialization)}}
  When restricting our attention, as we do here, to only first-order reversible
  functional programming languages, another example of inverse arrows arises in
  the form of \emph{serializers}. A serializer is a function that transforms an
  internal data representation into one more suitable for storage, or for
  transmission to other running processes. To transform serialized data back
  into an internal representation, a suitable deserializer is used. 
  
  When restricting ourselves to the first-order case, it seems reasonable to
  assume that all types are serializable, as we thus avoid the problematic 
  case of how to serialize data of function type. As such, assuming that all 
  types \<X\> admit a function \<serialize : X\leftrightarrow{Serialized} X\>
  (where \<Serialized X\> is the type of serializations of data of type X), we
  define the \<Serializer\> arrow as follows:
\begin{haskell*}
\hskwd{type} Serializer X Y = 
X\leftrightarrow{Serialized} Y\\\\\\
\hskwd{instance} Arrow (Serializer) \hskwd{where} \\
\quad\hsalign{
arr f x &=&\ serialize (f x) \\
(a\acmp{b}) x &=&\ b (serialize^\dagger (a x)) \\
first a (x,z) &=&\ serialize (serialize^\dagger (a x), z)
} \\\\\\
\hskwd{instance} InverseArrow (Serializer) \hskwd{where} \\
\quad\hsalign{
inv a y &=&\ serialize (a^\dagger (serialize y))
}
\end{haskell*}
  Notice how \<serialize^\dagger : Serialized X\leftrightarrow{X}\> takes the
  role of a (partial) deserializer, able to recover the internal representation
  from serialized data as produced by the serializer. A deserializer of the
  form \<serialize^\dagger\> will often only be partially defined, since many
  serialization methods allow many different serialized representations of the
  same data (for example, many textual serialization formats are whitespace
  insensitive). In spite of this shortcoming, partial deserializers produced by 
  inverting serializers are sufficient for the above definition to satisfy the 
  inverse arrow laws.
\end{example}

\begin{example}{\emph{(Dagger Frobenius monads)}}\label{ex:frobeniusmonads}
  Monads are also often used to capture computational side-effects. Arrows are more general.
  If $T$ is a strong monad, then $A=\hom(-,T(+))$ is an arrow: $\arr$ is given by the unit, $\acmp$ is given by Kleisli composition, and $\first$ is given by the strength maps.
  What happens when the base category is a dagger or inverse category modelling reversible pure functions?

  A monad $T$ on a dagger category is a \emph{dagger Frobenius monad} when it satisfies $T(f^\dag)=T(f)^\dag$ and $T(\mu_X) \circ \mu_{T(X)}^\dag = \mu_{T(X)} \circ T(\mu_X^\dag)$. 
  The Kleisli category of such a monad is again a dagger category~\cite[Lemma~6.1]{heunenkarvonen:daggermonads}, giving rise to an operation $\inv$ satisfying~\eqref{eq:daggerarrow1}--\eqref{eq:daggerarrow3}. 
  A dagger Frobenius monad is strong when the strength maps are unitary. In this case~\eqref{eq:daggerarrow4} also follows.
  If the underlying category is an inverse category, then $\mu \circ \mu^\dag \circ \mu = \mu$, whence $\mu \circ \mu^\dag = \id$, and~\eqref{eq:inversearrow1}--\eqref{eq:inversearrow2} follow. Thus, if $T$ is a strong dagger Frobenius monad on a dagger/inverse category, then $A$ is a dagger/inverse arrow.
  The Frobenius monad $T(X)=X \otimes \mathbb{C}^2$ on the category of Hilbert spaces captures measurement in quantum computation~\cite{heunenkarvonen:reversiblemonads}, giving a good example of capturing an irreversible effect in a reversible setting. For more examples see~\cite{heunenkarvonen:daggermonads}.  
\end{example}

\begin{example}{\emph{(Restriction monads)}}
  There is a notion in between the dagger and inverse arrows of the previous example.
  A \emph{(strong) restriction monad} is a (strong) monad on a (monoidal) restriction category whose underlying endofunctor is a restriction functor.
  The Kleisli-category of a restriction monad $T$ has a natural restriction structure: just define the restriction of $f\colon X\to T(Y)$ to be $\eta_X\circ\bar{f}$. The functors between the base category and the Kleisli category then become restriction functors.
  If $T$ is a strong restriction monad on a monoidal restriction category \cat{C}, then $\Inv(\cat{C})$ has an inverse arrow $(X,Y)\mapsto (\Inv(\Kl(T)))(X,Y)$.
\end{example}

\begin{example}{\emph{(Control flow)}}\label{ex:control_flow}
  While only trivial inverse categories have coproducts~\cite{giles:thesis},
  less structure suffices for reversible control structures. When the domain
  and codomain of an inverse arrow both have disjointness tensors (see
  Definition~\ref{def:disjtensor}), it can often be used to implement
  $\mathit{ArrowChoice}$. For a simple example, the pure arrow on an inverse
  category with disjointness tensors implements \<left : A X Y\to{A}
  (X\oplus{Z}) (Y\oplus{Z})\> as
  \begin{haskell}
    left f (x, z) = (f x, z)
  \end{haskell}
  The laws of \<ArrowChoice\>~\cite{hughes:arrows} simply reduce to $- \oplus
  -$ being a bifunctor with natural quasi-injections. More generally, the laws
  amount to preservation of the disjointness tensor. For the reversible state
  arrow (Example~\ref{ex:reversiblestate}), this hinges on $\otimes$
  distributing over $\oplus$.
  
  The splitting combinator $(+\!\!\!\!+\!\!\!\!+)$ is unproblematic for reversiblity, but the fan-in combinator $(|||)$ cannot be
  defined reversibly, as it explicitly deletes information about which branch 
  was chosen. Reversible conditionals thus require two 
  predicates: one determining the branch to take, and one asserted to join the branches after execution. The branch-joining predicate must be chosen carefully to ensure that it is always true after the
  \emph{then}-branch, and false after the \emph{else}-branch. This
  is a standard way of handling branch joining
  reversibly~\cite{yokoyamaglueck:janus,yokoyamaetal:rfun,glueckkaarsgaard:rfcl}.
\end{example}

\begin{example}{\emph{(Superoperators)}}\label{ex:cpm}
  Quantum information theory has to deal with environments.
  The basic category $\cat{FHilb}$ is that of finite-dimensional Hilbert spaces and linear maps. 
  But because a system may be entangled with its environment, the only morphisms that preserve states are the so-called {superoperators}, or \emph{completely positive} maps~\cite{selinger:dagger,coeckeheunen:completepositivity}: they are not just positive, but stay positive when tensored with an arbitrary ancillary object. In a sense, information about the system may be stored in the environment without breaking the (reversible) laws of nature. This leads to the so-called CPM construction. It is infamously known \emph{not} to be a monad. But it \emph{is} a dagger arrow on $\cat{FHilb}$, where $A~X~Y$ is the set of completely positive maps $X^* \otimes X \to Y^* \otimes Y$, $\arr f = f_* \otimes f$, $a \acmp b = b \circ a$, $\first_{X,Y,Z} a = a \otimes \id[Z^* \otimes Z]$, and $\inv a = a^\dag$.
\end{example}

Aside from these, other examples do fit the interface of inverse arrows, though
they are less syntactically interesting as they must essentially be ``built
in'' to a particular programming language. These include reversible IO, which
functions very similarly to irreversible IO, and reversible recursion, which
could be used to give a type-level separation between terminating and
potentially non-terminating functions, by only allowing fixed points of
parametrized functions between arrows rather than between (pure) functions.

\section{Inverse arrows, categorically}\label{sec:arrowscategorically}

This section explicates the categorical structure of inverse arrows. Arrows on $\cat{C}$ can be modelled categorically as monoids in the functor category $\prof$~\cite{jacobs2009categorical}. They also correspond to certain identity-on-objects functors $J\colon \cat{C}\to\cat{D}$. The category $\cat{D}$ for an arrow $A$ is built by $\cat{D}(X,Y)=A~X~Y$, and $\arr$ provides the functor $J$. We will only consider the multiplicative fragment, except for remark~\ref{rem:strength}. The operation $\first$ can be incorporated in a standard way using strength~\cite{jacobs2009categorical,asada2010arrows}, and poses no added difficulty in the reversible setting.


Clearly, dagger arrows correspond to $\cat{D}$ being a dagger category and $J$ a dagger functor, whereas inverse arrows correspond to both $\cat{C}$ and $\cat{D}$ being inverse categories and $J$ a (dagger) functor. This section addresses the following question: which monoids correspond to dagger arrows and inverse arrows? In the dagger case, the answer is quite simple: the dagger makes $\prof$ into an involutive monoidal category, and then dagger arrows correspond to involutive monoids. Inverse  arrows furthermore require certain diagrams to commute. 

\begin{definition} 
  An \emph{involutive monoidal category} is a monoidal category $\cat{C}$ equipped with an \emph{involution}: a functor $\overline{(\ )}\colon \cat{C}\to \cat{C}$  satisfying $\overline{\overline{f}}=f$ for all morphisms $f$, together with a  natural isomorphism
  $
    \chi_{X,Y}\colon \overline{X}\otimes\overline{Y}\to \overline{Y\otimes X} 
  $
  that makes the following diagrams commute\footnote{There is a more general definition allowing a natural isomorphism $\overline{\overline{X}}\to X$ (see~\cite{egger2011involutive} for details), but we only need the strict case.}:
  \[
      \begin{aligned}\begin{tikzpicture}
        \matrix (m) [matrix of math nodes,row sep=2em,column sep=4em,minimum width=2em]
          {\overline{X}\otimes(\overline{Y}\otimes \overline{Z}) & (\overline{X}\otimes\overline{Y})\otimes \overline{Z} \\
            \overline{X}\otimes \overline{Z\otimes Y} &\overline{Y\otimes X}\otimes \overline{Z} \\
            \overline{(Z\otimes Y)\otimes X} & \overline{Z\otimes(Y\otimes X)} \\};
          \path[->]
          (m-1-1) edge node [left] {$\id\otimes\chi$} (m-2-1)
                edge node [above] {$\alpha$} (m-1-2)
          (m-1-2) edge node [right] {$\chi\otimes\id$} (m-2-2)
          (m-2-2) edge node [right] {$\chi$} (m-3-2)
          (m-3-2) edge node [below] {$\overline{\alpha}$} (m-3-1)
          (m-2-1) edge node [left] {$\alpha$} (m-3-1);
      \end{tikzpicture}\end{aligned}
      \qquad\qquad
      \begin{aligned}\begin{tikzpicture}
          \matrix (m) [matrix of math nodes,row sep=2em,column sep=4em,minimum width=2em]
          {\overline{\overline{X}}\otimes \overline{\overline{Y}}  & \overline{\overline{Y}\otimes\overline{X}}  \\
            X\otimes Y& \overline{\overline{X\otimes Y}}  \\};
          \path[->]
          (m-1-1) edge node [left] {$\id$} (m-2-1)
                edge node [above] {$\chi$} (m-1-2)
          (m-1-2) edge node [right] {$\overline{\chi}$} (m-2-2)
          (m-2-1) edge node [below] {$\id$} (m-2-2);
      \end{tikzpicture}\end{aligned}
  \]
\end{definition}

Just like monoidal categories are the natural setting for monoids, involutive monoidal categories are the natural setting for involutive monoids. Any involutive monoidal category has a canonical isomorphism $\phi\colon I\to\overline{I}$~\cite[Lemma~2.3]{egger2011involutive}:
\[\begin{tikzpicture}[xscale=3]
  \node (1) at (0,0) {$I=\overline{\overline{I}}$};
  \node (2) at (1,0) {$\overline{\overline{I} \otimes I}$};
  \node (3) at (2,0) {$\overline{I} \otimes \overline{\overline{I}} = \overline{I} \otimes I$};
  \node (4) at (3,0) {$\overline{I}$};
  \draw[->] (1) to node[above]{$\overline{\rho_{\overline{I}}}^{-1}$} (2);
  \draw[->] (2) to node[above]{$\chi_{I,\overline{I}}^{-1}$} (3);
  \draw[->] (3) to node[above]{$\rho_{\overline{I}}$} (4);
\end{tikzpicture}\]
Moreover, any monoid $M$ with multiplication $m$ and unit $u$ induces a monoid on $\overline{M}$ with multiplication $\overline{m}\circ \chi_{M,M}$ and unit $\overline{u}\circ\phi$. This monoid structure on $\overline{M}$ allows us to define involutive monoids.

\begin{definition}
  An \emph{involutive monoid} is a monoid $(M,m,u)$ together with a monoid homomorphism $i\colon \overline{M}\to M$ satisfying $i\circ\overline{i}=\id$. A \emph{morphism} of involutive monoids is a monoid homomorphism $f\colon M\to N$ making the following diagram commute:
  \[
      \begin{aligned}\begin{tikzpicture}
        \matrix (m) [matrix of math nodes,row sep=1.5em,column sep=4em,minimum width=2em]
        {\overline{M}  & \overline{N}  \\
          M& N  \\};
        \path[->]
        (m-1-1) edge node [left] {$i_M$} (m-2-1)
            edge node [above] {$\overline{f}$} (m-1-2)
        (m-1-2) edge node [right] {$i_N$} (m-2-2)
        (m-2-1) edge node [below] {$f$} (m-2-2);
      \end{tikzpicture}\end{aligned}
  \]
\end{definition}

Our next result lifts the dagger on $\cat{C}$ to an involution on the category $[\cat{C}\op \times \cat{C},\cat{Set}]$ of profunctors. First we recall the monoidal structure on that category. It categorifies the dagger monoidal category $\cat{Rel}$ of relations of Section~\ref{sec:inversecategories}~\cite{borceux}.
\begin{definition}
  If \cat{C} is small, then $\prof$ has a monoidal structure
      \begin{equation*}
        F\otimes G (X,Z)= \int^Y F(X,Y)\times G(Y,Z)\text;
      \end{equation*}
  concretely, $F\otimes G (X,Z)= \coprod_{Y\in \cat{C}} F(X,Y)\times G(Y,Z)/\approx$, 
  where $\approx$ is the equivalence relation generated by $(y,F(f,\id)(x))\approx (G(\id,f)(y),x)$, and the action on morphisms is given by $F\otimes G (f,g):= [y,x]_\approx \mapsto [F(f,\id)x,G(\id,g)y]$. The unit of the tensor product is $\hom_\cat{C}$.
\end{definition}

\begin{proposition}\label{prop:swapdaggeronprof}
  If \cat{C} is a dagger category, then $\prof$ is an involutive monoidal category when one defines the involution on objects $F$ by $\overline{F}(X,Y) = F(Y,X)$, $\overline{F}(f,g)=F(g^\dag,f^\dag)$ and on morphisms $\tau\colon F\to G$ by $\overline{\tau}_{X,Y}=\tau_{Y,X}$.
\end{proposition}
\begin{proof}
  First observe that $\overline{(\ )}$ is well-defined: For any natural transformation of profunctors $\tau$, $\overline{\tau}$ is natural, and $\tau\mapsto \overline{\tau}$ is functorial.
  Define $\chi_{F,G}$ by the following composite of natural isomorphisms: 
    \begin{align*}
      \overline{F}\otimes \overline{G} (X,Z)
      &\cong \textstyle\int^Y \overline{F}(X,Y)\times\overline{G}(Y,Z)\text{ by definition of }\otimes \\
      &= \textstyle\int^Y F(Y,X)\times G(Z,Y)\text{ by definition of }\overline{(\ )} \\
      &\cong \textstyle\int^Y G(Z,Y)\times F(Y,X) \text{ by symmetry of }\times\\
      &\cong G\otimes F (Z,X)\text{ by definition of }\otimes \\
      &=\overline{G\otimes F} (X,Z)\text{ by definition of }\overline{(\ )}
     \end{align*}
  Checking that $\chi$ make the relevant diagrams commute is routine.
\end{proof}

\begin{theorem} 
  If \cat{C} is a dagger category, the multiplicative fragments of dagger arrows on \cat{C} correspond exactly to involutive monoids in \prof.
\end{theorem}
\begin{proof} 
  It suffices to show that the dagger on an arrow corresponds to an involution on the corresponding monoid $F$. But this is easy: an involution on $F$ corresponds to giving, for each $X,Y$ a map $F(X,Y)\to F(Y,X)$ subject to some axioms. That this involution is a monoid homomorphism amounts to it being a contravariant identity-on-objects-functor, and the other axiom amounts to it being involutive.
\end{proof}

\begin{remark}\label{rem:strength}
  If the operation $\first$ is modeled categorically as (internal) strength, axiom~\eqref{eq:daggerarrow4} for dagger arrows can be phrased in $\prof$ as follows: for each object $Z$ of $\cat{C}$, and each dagger arrow $M$, the profunctor $M_Z=M((-)\otimes Z,(+)\otimes Z)$ is also a dagger arrow, and $\first_{-,+,Z}$ is a natural transformation $M\Rightarrow M_Z$. The arrow laws~\eqref{eq:arrow7} and~\eqref{eq:arrow8} imply that it is a monoid homomorphism, and the new axiom just states that it is in fact a homomorphism of involutive monoids. For inverse arrows this law is not needed, as any functor between inverse categories is automatically a dagger functor and thus every monoid homomorphism between monoids corresponding to inverse arrows preserves the involution.
\end{remark}

Next we set out to characterize which involutive monoids correspond to inverse arrows. Given an involutive monoid $M$, the obvious approach would be to just state that  the map $M\to M$ defined by $a\mapsto a \circ a^\dag \circ a$ is the identity. However, there is a catch: for an arbitrary involutive monoid, the map $a\mapsto a \circ a^\dag \circ a$ is not natural transformation and therefore not a morphism in $\prof$. To circumvent this, we first require some conditions guaranteeing naturality. These conditions concern endomorphisms, and to discuss them we introduce an auxiliary operation on $\prof$. \cut{It categorifies the operation of a restriction category of Definition~\ref{def:restrictioncategory} to profunctors on dagger categories. }

\begin{definition} 
  Let $\cat{C}$ be a dagger category. Given a profunctor $M \colon \cat{C}\op \times \cat{C} \to\cat{Set}$, define $LM \colon \cat{C}\op \times \cat{C} \to \cat{Set}$ by 
  \begin{align*}
    LM(X,Y)&=M(X,X)\text, \\
    LM(f,g)&=f^\dag \circ (-) \circ f\text.
  \intertext{If $M$ is an involutive monoid in $\prof$, define a subprofunctor of $LM$:}
    L^+M(X,Y) &= \{a^\dag \circ a \in M(X,X) \mid a \in M(X,Z)\text{ for some }Z\}\text.
  \end{align*} 
\end{definition}

\begin{remark}
  The construction $L$ is a functor $\prof \to \prof$. 
  There is an analogous construction $RM(X,Y)=M(Y,Y)$ and $R^+M$, and furthermore $RM=\overline{LM}$.
  For any monoid $M$ in $\prof$, $LM$ is a right $M$-module (and $RM$ a left $M$-module).
  Compare Example~\ref{ex:cpm}.
\end{remark}

For the rest of this section, assume the base category \cat{C} to be an inverse category. This lets us multiply positive arrows by positive pure morphisms. If $M$ is an involutive monoid in $\prof$, then the map $LM\times L^+(\hom_\cat{C})\to LM$ defined by $(a,g^\dag\circ g)\mapsto a\circ g^\dag \circ g$ is natural:
\begin{align*}
    &LM\times L^+(\hom )(f,\id[Y])(a,g^\dag \circ g) \\
    &=(f^\dag \circ a \circ f,f^\dag \circ g^\dag \circ g\circ f) \\
    &\mapsto f^\dag\circ a\circ f\circ f^\dag \circ g^\dag \circ g\circ f \\
    &=f^\dag\circ a\circ g^\dag\circ g\circ f\circ f^\dag\circ f&\text{ because \cat{C} is an inverse category}\\
    &=f^\dag \circ a\circ g^\dag\circ g\circ f&\text{ because \cat{C} is an inverse category} \\
    &=LM(f,\id[Y])(a\circ g^\dag\circ g)
\end{align*}

Similarly there is a map $L^+(\hom)\times LM\to LM$ defined by $(g^\dag \circ g,a)\mapsto g^\dag\circ g\circ a$. Now the category corresponding to $M$ satisfies $a^\dag \circ a \circ g^\dag\circ  g=g^\dag \circ g\circ  a^\dag\circ  a$ for all $a$ and pure $g$ if and only if the following diagram commutes:
\begin{equation}\label{diag:step0}
    \begin{aligned}\begin{tikzpicture}
        \matrix (m) [matrix of math nodes,row sep=1.5em,column sep=4em,minimum width=2em]
        {L^+M\times L^+(\hom)  && LM\times L^+(\hom)  \\
          L^+(\hom)\times L^+M&   L^+(\hom)\times LM& LM\\};
        \path[->]
        (m-1-1) edge node [left] {$\sigma $} (m-2-1) 
            edge node [above] {$$} (m-1-3) 
        (m-1-3) edge node [right] {$$} (m-2-3) 
        (m-2-1) edge node [below] {$$} (m-2-2) 
        (m-2-2) edge node [below] {$$} (m-2-3); 
    \end{tikzpicture}\end{aligned}
\end{equation}
If this is satisfied for an involutive monoid $M$ in $\prof$, then positive arrows multiply. In other words, the map $L^+M\times L^+M\to LM$ defined by $(a^\dag \circ a,b^\dag \circ b)\mapsto a^\dag\circ a\circ b^\dag\circ b$ is natural:
\begin{align*}
  &D_M(f,g)(a,a^\dag,a)\\
  &=(g \circ a\circ f,f^\dag \circ a^\dag \circ g^\dag,g\circ a\circ f) \\
  &\mapsto g\circ a\circ f\circ f^\dag\circ  a^\dag\circ  g^\dag \circ g\circ a\circ f \\
  &=g\circ a\circ a^\dag\circ  g^\dag\circ  g\circ a\circ  f\circ f^\dag\circ  f &\text{ by~\eqref{diag:step0}}\\
  &=g\circ a\circ a^\dag\circ  g^\dag\circ  g\circ a\circ f&\text{ because \cat{C} is an inverse category} \\
  &=g\circ g^\dag \circ g\circ  a\circ a^\dag \circ a\circ  f &\text{ by~\eqref{diag:step0}} \\
  &=g\circ a\circ a^\dag\circ  a\circ  f&\text{ because \cat{C} is an inverse category} \\
  &=M(f,g)(a\circ a^\dag\circ  a)
  \end{align*}
This multiplication is commutative iff the following diagram commutes:
\begin{equation}\label{diag:step1 (positives commute)}
    \begin{aligned}\begin{tikzpicture}
        \matrix (m) [matrix of math nodes,row sep=2em,column sep=4em,minimum width=2em]
        {L^+M\times L^+M  & L^+M\times L^+M  \\
          & LM  \\};
        \path[->]
        (m-1-1) edge node [left] {$$} (m-2-2)
            edge node [above] {$\sigma$} (m-1-2)
        (m-1-2) edge node [right] {$$} (m-2-2);
    \end{tikzpicture}\end{aligned}
\end{equation}

Finally, let $D_M\hookrightarrow M\times\overline{M}\times M$ be the diagonal
$
  D_M(X,Y)=\{(a,a^\dag,a)\mid a\in M(X,Y)\}
$.

If $M$ satisfies~\eqref{diag:step0}, then the map $D_M\to M$ defined by $(a,a^\dag, a)\mapsto a\circ a^\dag\circ a$ is natural:

\begin{align*}
  &D_M(f,g)(a,a^\dag,a)\\
  &=(g \circ a\circ f,f^\dag \circ a^\dag \circ g^\dag,g\circ a\circ f) \\
  &\mapsto g\circ a\circ f\circ f^\dag\circ  a^\dag\circ  g^\dag \circ g\circ a\circ f \\
  &=g\circ a\circ a^\dag\circ  g^\dag\circ  g\circ a\circ  f\circ f^\dag\circ  f &\text{ by~\eqref{diag:step0}}\\
  &=g\circ a\circ a^\dag\circ  g^\dag\circ  g\circ a\circ f&\text{ because \cat{C} is an inverse category} \\
  &=g\circ g^\dag \circ g\circ  a\circ a^\dag \circ a\circ  f &\text{ by~\eqref{diag:step0}} \\
  &=g\circ a\circ a^\dag\circ  a\circ  f&\text{ because \cat{C} is an inverse category} \\
  &=M(f,g)(a\circ a^\dag\circ  a)
  \end{align*}

 Thus $M$ satisfies $a\circ a^\dag\circ a=a$ if and only if the following diagram commutes:
\begin{equation}\label{diag:step2 (aa^daga=a)}
    \begin{aligned}\begin{tikzpicture}
        \matrix (m) [matrix of math nodes,row sep=2em,column sep=4em,minimum width=2em]
        {M  & D_M \\
          & M \\};
        \path[->]
        (m-1-1) edge node [below] {$\id$} (m-2-2)
            edge node [above] {$$} (m-1-2)
        (m-1-2) edge node [right] {$$} (m-2-2);
      \end{tikzpicture}\end{aligned}
\end{equation}
Hence we have established the following theorem. 

\begin{theorem}\label{thm:characterizing_inverse_arrows}
  Let \cat{C} be an inverse category. Then the multiplicative fragments of inverse arrows on \cat{C} correspond exactly to involutive monoids in $\prof$ making the diagrams~\eqref{diag:step0}--\eqref{diag:step2 (aa^daga=a)} commute.
  \qed
\end{theorem}

\section{Applications and related work}
As we have seen, inverse arrows capture a variety of fundamental reversible
effects. An immediate application of our results would be to retrofit existing
typed reversible functional programming languages (\eg,
Theseus~\cite{jamessabry:theseus}) with inverse arrows to accommodate 
reversible effects while maintaining a type-level separation between pure and
effectful programs. Another approach could be to design entirely new such
programming languages, taking inverse arrows as the fundamental representation
of reversible effects.
While the Haskell approach to arrows uses typeclasses~\cite{hughes:arrows},
these are not a priori necessary to reap the benefits of inverse arrows. For
example, special syntax for defining inverse arrows could also be used, either
explicitly, or implicitly by means of an effect system that uses inverse arrows
``under the hood''.

\cut{To aid programming with ordinary arrows, a handy notation due to
Paterson~\cite{paterson:notation,paterson:computation} may be used. If the underlying monoidal dagger category has natural coassociative diagonals, for example when it has inverse products, a similar $\mathtt{do}$-notation can be implemented for inverse and dagger arrows.}

To aid programming with ordinary arrows, a handy notation due to
Paterson~\cite{paterson:notation,paterson:computation} may be used. The
simplest form of this notation is based on process combinators, the central one
being
\begin{equation*}
  p \to e_1 \prec e_2 = \left\{\begin{array}{ll}
    \mathit{arr}(\lambda p. e_2) \acmp e_1 & \quad\text{if $p$ is fresh for $e_1$,} \\
    \mathit{arr}(\lambda p. (e_1,e_2)) \acmp \text{app} & \quad\text{otherwise.}
  \end{array} \right.
\end{equation*}
Note that if the second branch is used, the arrow must
additionally be an instance of $\mathit{ArrowApply}$ (so that it is, in fact,
a monad). Though we only know of degenerate examples where inverse arrows are
instances of $\mathit{ArrowApply}$, this definition is conceptually
unproblematic (from the point of view of guaranteeing reversibility) so long as the pure function $\lambda p. e_2$ is first-order and
reversible. A more advanced style of this notation is the
\emph{do}-notation for arrows, which additionally relies on the arrow
combinator
\begin{haskell*}
  bind &:& A~X~Y\to{A~(X\otimes{Y})~Z}\to{A~X~Z} \\
  f~'bind'~g &=& (arr(id)\afanout{f})\acmp{g}\enspace.
\end{haskell*}
If the underlying monoidal dagger category has natural coassociative diagonals,
for example when it has inverse products, this combinator does exist: the arrow
combinator $(\afanout)$ can be defined as
\begin{haskell*}
  (\afanout) &:& A~X~Y\to{A~X~Z}\to{A~X~(Y\otimes{Z})} \\
  f\afanout{g} &=& arr(copy) \acmp first(f) \acmp second(g)
\end{haskell*}
where \<copy : X\to{X}\otimes{X}\> is the natural diagonal (given in pseudocode
by \<copy x = (x,x)\>), and the combinator \<second\> is derived from \<first\> in the
usual way, \ie, as
\begin{haskell*}
  second &:& A~X~Y\to{A}~(Z\otimes{X})~(Z\otimes{Y}) \\
  second f &=& arr(swap)\acmp{first(f)}\acmp{arr(swap)}
\end{haskell*}
with \<swap : X\otimes{Y}\leftrightarrow{Y}\otimes{X}\> given by \<swap (x,y) = (y,x)\>. This allows \emph{do}-notation of the form
\begin{haskell*}
  \hskwd{do} \{p\leftarrow{c}\scolon{A}\} \equiv c 'bind' 
  (\kappa{p}. \hskwd{do} \{A\})\text,
\end{haskell*}
so soon as the $\kappa$-calculus~\cite{hasegawa:kappa_calc}
expression \<\kappa{p}. \hskwd{do} \{A\}\> is reversible. Note, however, that
\emph{do}-expressions of the form \<\hskwd{do} \{c\scolon A\}\> (\ie, where 
the output of $c$ is discarded entirely) will fail to be reversible in all but
the most trivial cases. 
Since \<\hskwd{do} \{p \leftarrow c\scolon A\}\> produces a value of an inverse arrow type, closure under program inversion provides a program we might call
\begin{haskell*}
  \hskwd{undo} \{p\leftarrow{c}\scolon{A}\} \equiv inv(\hskwd{do} 
  \{p\leftarrow{c}\scolon{A}\}) \enspace\text.
\end{haskell*}
Inverse arrow law~\eqref{eq:inversearrow1} then guarantees that \emph{do}ing, then \emph{undo}ing, and then \emph{do}ing the same operation is the same as \emph{do}ing it once.

A pleasant consequence of the semantics of inverse arrows is that
inverse arrows are safe: as long as the inverse arrow laws are satisfied, fundamental properties guaranteed by reversible functional
programming languages (such as invertibility and closure under program
inversion) are preserved. In this way, inverse arrows provide reversible
effects as a conservative extension to pure reversible functional programming.

A similar approach to invertibility using arrows is given by bidirectional
arrows~\cite{alimarineetal:biarrows}. However, while the goal of inverse arrows
is to add effects to already invertible languages, bidirectional arrows arise
as a means to add invertibility to an otherwise uninvertible language. As such,
bidirectional arrows have different concerns than inverse arrows, and notably
do not guarantee invertibility in the general case.

\subsubsection*{Acknowledgements}
This work was supported by COST Action IC1405, the Oskar Huttunen Foundation, and EPSRC Fellowship EP/L002388/1. We thank Robert Furber and Robert Gl{\"u}ck for discussions.


\bibliographystyle{plainurl}
\bibliography{reversiblearr}

\end{document}